\documentclass[journal,twoside,web]{ieeecolor}
\usepackage{lcsys}
\usepackage{cite}
\usepackage{caption}
\usepackage{amsmath,amssymb,amsfonts}
\usepackage{algorithmic}
\usepackage{graphicx}
\usepackage{textcomp}
\usepackage{bm}
\usepackage{soul}
\usepackage{nicefrac}

\usepackage{babel,blindtext}
\usepackage{mathtools}  
\usepackage{commath}
\usepackage{array}
\usepackage{derivative}
\usepackage{enumerate}
\usepackage[colorlinks,allcolors=black]{hyperref} 
\usepackage[capitalize,nameinlink]{cleveref}
\DeclareMathOperator{\sign}{sgn}
\usepackage{MnSymbol}
\usepackage{mathrsfs}
\newtheorem{theorem}{Theorem}
\newtheorem{lemma}{Lemma}

\newtheorem{defn}{Definition}

\newtheorem{remark}{Remark}
\newtheorem{assumption}{Assumption}

\usepackage[usenames,dvipsnames]{xcolor}
\usepackage{tikz}
\usetikzlibrary{positioning}
\usetikzlibrary{decorations.pathreplacing,calligraphy}

\usepackage{stackengine}

\usepackage[colorlinks,allcolors=black]{hyperref} 
\usepackage[capitalize,nameinlink]{cleveref}

\definecolor{blue}{rgb}{0.13, 0.67, 0.8}
\definecolor{purple}{rgb}{0.41, 0.16, 0.38}
\definecolor{black}{rgb}{0, 0, 0}

\usepackage{type1cm}
\usepackage{eso-pic}
\usepackage{color}

\makeatletter
\AddToShipoutPicture{%
            \setlength{\@tempdimb}{.5\paperwidth}%
            \setlength{\@tempdimc}{.5\paperheight}%
            \setlength{\unitlength}{1pt}%
            \put(\strip@pt\@tempdimb,\strip@pt\@tempdimc){%
        \makebox(0,0){\rotatebox{45}{\textcolor[gray]{0.92}%
        {\fontsize{6cm}{6cm}\selectfont{Preprint}}}}%
            }%
}

\begin{document}
\title{Approximation- and Chattering-free Quasi Sliding Mode Control for Unknown Systems}
\author{Pankaj K Mishra  and Pushpak Jagtap, \IEEEmembership{Member, IEEE}
\thanks{This paper is submitted for consideration of publication.}
\thanks{This work was partly supported by the Google Research Grant, the SERB Start-up Research Grant, and the CSR Grant by Nokia on Networked Robotics.}
\thanks{The authors are with the Robert Bosch Centre for Cyber-Physical Systems,
Indian Institute of Science, Bangalore 560012, India (e-mail:
\{pankajmishra,pushpak\}@iisc.ac.in)}}
\maketitle
\thispagestyle{empty} 
\begin{abstract}

Motivated by the concept of Quasi-Sliding Mode (QSM) in discrete-time systems, this paper presents a novel approach that relaxes the requirement of ubiquitous exact sliding motion in continuous-time systems, aiming to achieve an approximation- and chattering-free quasi-sliding mode controller (QSMC). The proposed QSMC provides robust and chattering-free control for an unknown nonlinear system without the need for intricate control methods, learning agents, or system identification tools. Furthermore, comprehensive simulation studies have been conducted to validate the effectiveness of the proposed strategy.
\end{abstract}

\section{INTRODUCTION}

Sliding mode control (SMC) has emerged as a prominent field of research for its numerous benefits, including rapid convergence and robustness against parametric uncertainty and external disturbances \cite{smcbook}. However, SMC suffers from two significant limitations in real-time applications: ``chattering" or high-frequency oscillation in control input and the challenge of estimating unknown system parameters for equivalent control. Here, by equivalent control, we refer to a control law designed for an unknown system, considering the system is known. Although significant progress has been made in mitigating chattering through complex control laws utilizing higher order SMC and nonlinear sliding surface\cite{FENG20141310, DU201687, 8064665, 880629, 7913590}, the issue of dealing with unknown system parameters for equivalent control continues to be a persistent challenge. Some works have addressed this challenge  \cite{1263584,8744495, 9942821}. For instance, in \cite{1263584}, two parallel neural networks (NNs) were used for a class of linear systems to tackle the chattering problem and unknown parameters in equivalent control. In \cite{8744495}, a complex recurring NN with two hidden layers were employed for a class of nonlinear systems to handle unknown system parameters, albeit overlooking the chattering issue. In \cite{9942821}, a combination of recursive NN with two hidden layers and a fuzzy neural network was used to approximate the equivalent control computed using nonlinear recursive sliding mode strategy \cite{8818641}. While this framework helped mitigate unknown system parameter estimation and reduce chattering, complete elimination of chattering was not achieved, and the control strategy was limited to a specific class of second-order nonlinear systems.\par 
Considering the aforementioned works, higher-order SMC or nonlinear sliding surfaces with complex learning agents seem to be a plausible approach to address the aforementioned issues in SMC. However, using learning tools such as fuzzy logic, NNs, or adaptive algorithms for system identification or parameter estimation comes with computational complexity, many hyperparameters, and challenges in designing complex adaptive laws and generating persistently exciting input signals. Moreover, the results of such tools are typically limited to compact sets, and determining such sets and guaranteeing that the system trajectory remains within them can be challenging, often leading to ad-hoc assumptions. This further adds complexity to the overall approach and may impact the robustness, thus compromising the advantage of SMC. Motivated by the ideas presented in \cite{10004950,10093463}, we propose a Quasi Sliding Mode (QSM) control approach (which has its roots in discrete-time systems) that aims at achieving approximation- and chattering-free control.\par
Traditionally, SMC is designed in two steps: $i)$ the selection of a sliding surface so that the system trajectory has desirable behaviour when confined to it, and $ii)$ the design of control law such that   the trajectory  of the system intersects or reaches the sliding surface and then holds them there so that trajectory slides along it.
\par Looking at the state-of-the-art literature on SMC, it can be found that achieving sliding mode in discrete-time has  remained an issue and leads to state-of-the-art term Quasi Sliding Mode (QSM) \cite{370376, 704892, 4132912}.  Nicely pointed out in \cite{4132912} that the control input in discrete-time sliding mode control is only calculated once per sampling interval and remains fixed throughout the interval. As a result, the system state trajectory may not move along the sliding surface and resemble a sliding-like or quasi-sliding motion. In discrete-time systems, QSM is often a compulsion rather than a choice, as it can be challenging to achieve exact sliding motion. In contrast, obtaining sliding mode is straightforward and ubiquitous in continuous-time systems, making it a feasible option. So ideally achieving QSM in the continuous system typically does not  make sense. Thus so far in the literature, there are counted works \cite{4105940, LI2009103} on QSMC for the continuous-time system.  One should not confuse QSMC with quasi-continuous control, which is an extension of unit controller \cite{1102073,utkin2017sliding} to higher order sliding mode controller \cite{1532412, 8123916}. In \cite{4105940, LI2009103}  QSM based repetitive control   is developed for a system with known nominal dynamics. Repetitive control \cite{8734866} is a control scheme for tracking control and disturbance rejection of periodic signals. This paper is motivated by the reaching law  and time-varying sliding surface-based control strategy of  \cite{704892} and \cite{CHOI1994899,bartoszewicz1996time}, respectively, we have brought out the idea of QSMC for a continuous-time unknown nonlinear system.  We will see that by relaxing the  SMC, i.e., leveraging QSMC, we can eliminate the necessity of knowing the system's dynamics for designing the control law. Thus, it leads to an approximation-free QSMC. By approximation-free, we mean to say that it neither requires any learning agent like a neural network, fuzzy logic or any system identification tools nor any adaptive law to deal with the unknown dynamics of the system. The contributions of this paper are summarized below:
\begin{enumerate}
    \item We introduced the notion of Quasi Sliding Mode (QSM) for continuous-time systems, which facilitates the designs of approximation- and chattering-free controllers.
    \item To achieve QSM in a fixed time, a novel reaching and control law has been proposed.
\end{enumerate}
 
The remainder of this paper is organized as follows: In Section \ref{sec2}, preliminaries and problem statements are
presented.  Section \ref{sec3} presents the proposed QSMC. Section \ref{sec4} presents the stability analysis. Section \ref{sec5} presents the simulation results and discussion. Finally, Section \ref{sec6} concludes this paper.

The \textbf{notation} used in this paper includes $\mathfrak{L}^\infty$ for the set of all essentially bounded measurable functions. $\mathbb{R}$ and $\mathbb{N}$ denotes set of real and natural number. We use $\mathbb{R}^n$ to represent a vector in $n$-dimensional real space. In this paper, we use bold and normal letters to represent vector and scalar quantities, respectively.  For $x (t)\in \mathbb{R}$, $x\uparrow a$: $x$ approaches a real value $a$ from the left side,  $x\downarrow a$:  $x$ approaches a real value $a$ from the right side. $x^{(n)}$ represents the $n$th derivative of $x$. $\binom{m}{k}$ denotes binomial coefficients. $\bar{\mathbf{0}}_i$: $i$-dimensional zero vector.
 $$\sign{(x)}=\begin{cases}1, x>0\\
0, x=0\\
-1, x<0\end{cases}$$

\section{Preliminaries and Problem Statement}\label{sec2}
Consider a class of strict-feedback systems, 
\begin{equation}\label{sys}
    \begin{split}
        \dot x_i&=x_{i+1}, ~ i= 1,\ldots, n-1,\\
        \dot x_n&=f(\textbf{x})+g(\textbf{x})u+d(t),\\
        y&=x_1,
    \end{split}
\end{equation}
where $\textbf{x}=[x_1, \ldots, x_n]^T\in \mathbb{R}^n$ is a state vector, $f:\mathbb{R}^n\rightarrow \mathbb{R}$ is the unknown nonlinear function, $g:\mathbb{R}^n\rightarrow \mathbb{R}$ is unknown control coefficient, $d(t) \in \mathbb{R}$ is unknown piecewise continuous bounded disturbance, $u\in \mathbb{R} $  and $y\in \mathbb{R}$ are the input and output of the system, respectively. Let $x_{des}(t)\in \mathbb{R}$ be the smooth desired output, and $e=x_1-x_{des}$ be the  tracking error and $\textbf{e}=[e^{(0)},\ldots,e^{(n-1)}]^T$ a  vector of derivatives of error, where $e^{(i)}$ represents $i$th time-derivative. \par
In general, for tracking problems, a sliding surface can be described by  $\sigma(\textbf{e})=0$, where $\sigma$ is a sliding variable, and can be defined in error state space as $\sigma(\textbf{e};\textbf{c})=\textbf{c}^T\textbf{e}$ or
\begin{equation}\label{slide}
\sigma(\textbf{e};\textbf{c})=c_1e^{(0)}+c_2e^{(1)}+\cdots+c_{n-1}e^{(n-2)}+e^{(n-1)},
\end{equation}
with parameter vector $\textbf{c}=[c_1, \ldots, c_{n-1}, 1]^T$ is chosen such that the polynomial $s^{n-1}+c_{n-1}s^{n-2}+\cdots+c_{2}s+c_1$ is a Hurwitz polynomial, $s$ is a Laplace variable.\par
\textit{Note:}In this paper, instead of $\sigma(\textbf{e};\textbf{c})$, we will use  $\sigma(t)$ or simply $\sigma$ interchangeably depending on the context. This choice of notation facilitates our analysis and does not compromise the integrity or validity of our results.\par
Before moving to the problem statement, we will define the theme of this paper, i.e., quasi-sliding mode (QSM). As proposed in \cite{370376}, the traditional definition of quasi-sliding mode necessitated the system trajectory to cross the sliding surface in successive steps, resulting in undesirable chattering. However, a more refined definition introduced in \cite{704892} removes this requirement, allowing the system to remain within a small band around the sliding surface without crossing it in each control step. This advancement enables continuous control strategies, which offer smoother and more effective control actions, eliminating chattering and enhancing the overall system performance.  Motivated by this, we proposed our definition below.
\begin{defn}
    We say  system trajectory is in  ``\textbf{quasi-sliding mode} (QSM)"  if $|\sigma(t)|\le \varepsilon$ for all time $t>t_r$, where $\varepsilon$ as a small positive constant referred  as "\textbf{bandwidth of QSM}" and $t_r$ is \textbf{reaching time} for QSM.
\end{defn}

\textit{Control goal:} The paper aims to design a  controller $u$ for system \eqref{sys}, such that the output $y$ tracks the desired trajectory $y_d$. It can also be seen as a twofold goal. First, to design a reaching law for QSM. Second, to design an approximation-free QSMC tracking controller for \eqref{sys} so that the system trajectory achieves QSM and output track the desired trajectory.\par
To achieve the above goal, we have made the following assumptions.
\begin{assumption}\label{gassump}
The sign of control gain $g(\textbf{x})$ is known, and there exists an unknown positive constant $\underbar g$, such that $0< \underbar g\le |g(\textbf{x})|,$ for $\textbf{x}\in \mathbb{R}^n.$
\end{assumption}
\begin{assumption}\label{dassump}
There exist a positive constant $\bar d$ such that $|d(t)|<\bar d$, $\forall t\ge 0.$
\end{assumption}
\begin{assumption}\label{eassum}
    There exists a known positive constant $\bar e_i,$ such that, $|e^{(i)}(0)|<\bar e_i, ~i=0,\ldots, n-1.$ 
\end{assumption}
In the next section, we proposed a control strategy to satisfy the reaching condition to guarantee the reachability of the QSM.

\section{Proposed Quasi Sliding Mode Control}\label{sec3}
In our control approach, we aim to achieve Quasi-Sliding Mode (QSM) by first defining a reaching law. Drawing inspiration from the philosophical aspect of \cite{704892}, we build upon the concept of time-varying sliding surfaces \cite{CHOI1994899, bartoszewicz1996time} as the foundation for our reaching law. This law requires that the state of the controlled system remains within a predefined band around a time-varying sliding surface, which converges monotonically and in fixed time from its initial position to a pre-determined band. The \textbf{reaching law} to achieve $\varepsilon$ bandwidth of QSM is proposed as:
\begin{equation}\label{reach}
    |\sigma(t)|<\rho(t), \forall t\ge 0,
\end{equation}
where $\rho(t)=\rho_0\exp{(-\mu t)}+\rho_\infty$ with $\rho_0, \rho_\infty$ and $\mu$ are positive constants, which satisfy conditions $C1:\rho_0>|\sigma(0)|$ and $C2: \rho_\infty<\varepsilon<\rho_0$. It is a time-varying QSM reaching constraint which drives sliding variables to reach a predefined band $\varepsilon>\rho_\infty$, so that $|\sigma(t)|<\varepsilon, \forall t\ge t_r$ with $t_r<\frac{1}{\mu}\ln\left(\frac{\rho_0}{\varepsilon-\rho_\infty}\right).$ In Fig. \ref{intercept},  the graphical representation of reaching law has been shown  for the second order system or the case when $\sigma(t)=c_1e_1+e_2,$ with $e_1=e^{(0)}$ and $e_2=e^{(1)}$.

\begin{figure}[h]
    \centering
\includegraphics[trim=150 100 300 150,clip,width=12cm,height=8cm]{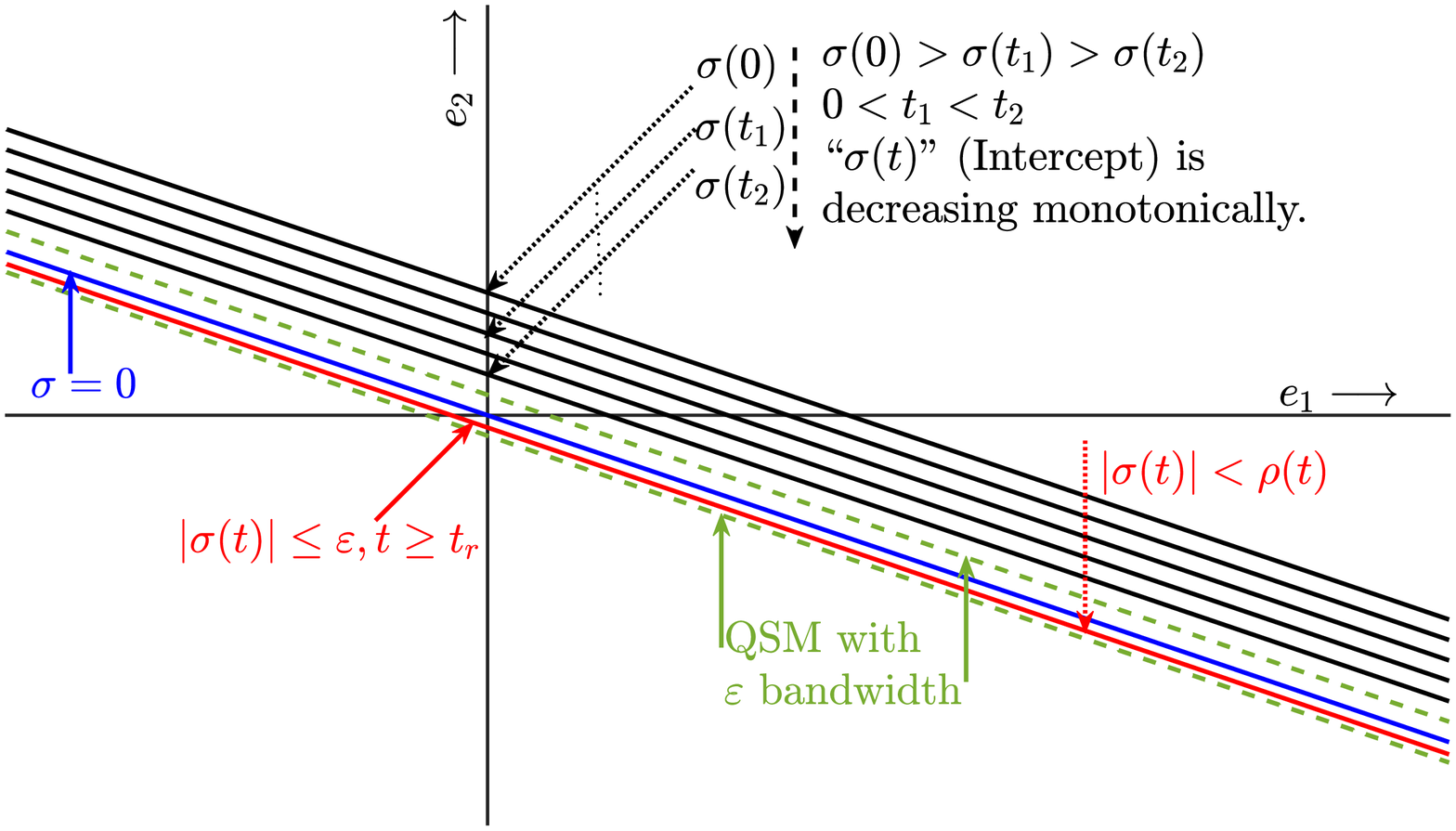}
\vspace*{-20mm}
    \caption{Graphical representation of reaching law for second order system ($\sigma=c_1e_1+e_2$); Plot of  $e_1$ vs $e_2$, with x-axis and y-axis intercept, $\sigma(t)/c_1$ and $\sigma(t)$, respectively. The black solid lines represent various hypersurfaces before reaching QSM with monotonically decreasing (red dotted line) intercepts of both axes. As per $\rho(t), $ in \eqref{reach}), eventually, the hypersurface will be inside the $\varepsilon$-bandwidth QSM (green dashed lines) in time $t\ge t_r$, shown as an arbitrary hypersurface (red solid line). The blue solid line represents the ideal sliding surface.}
    \label{intercept}
\end{figure}
The \textbf{control law} is proposed as 
\begin{equation}\label{input}
    u(t)=-\sign {(g(\textbf{x}))}\tan\left(\frac{\pi \sigma(t)}{2\rho(t)}\right).
\end{equation}

\begin{remark}  
By using Assumption \ref{gassump}, it can be concluded that the control input \eqref{input} is continuous, ensuring that it does not exhibit chattering. Moreover, it is noteworthy that \eqref{input} does not require any knowledge of the system, making it a low-complexity approximation- and chattering-free controller.
\end{remark}
Before we theorize that the control law will help to achieve the QSM, it is necessary to establish the following lemmas to assist its proof.

\begin{lemma}\label{lemma1}
If  the sliding variable  transgresses its upper bound, as given in \eqref{reach}, i.e., $``\rho(t)"$, then $(\sigma-\rho)$ will approach $0$ from the left side, and the following condition  will hold  \begin{align}
   \lim_{(\sigma-\rho)\uparrow 0}{\dot \sigma}&\ge-\mu\rho_{0},\label{taut1}
 \end{align}
\end{lemma}
\begin{proof}
Before exceeding its limits, the sliding variable must be within its bounds (i.e., $-\rho<\sigma<\rho$). This means that $ -2\rho<\sigma-\rho<0$ holds.
Thus if $\sigma$ is exceeding its upper bound, i.e., $\rho$, then $(\sigma-\rho)$ will approach $0$ from the left. Therefore, it is straightforward to deduce that when $(\sigma-\rho)$ approaches $0$ from the left, its time derivative will be greater than equal to $0$. As a result, we have
\begin{align}\label{zphi}
    \lim_{(\sigma-\rho)\uparrow 0}{\dot \sigma}\ge\dot\rho.
\end{align}
Noting that $\rho(t)=\rho_0\exp{(-\mu t)}+\rho_\infty$, so  we have $-\mu\rho_0\le\dot\rho\le0$. As a result, using \eqref{zphi} we have
 $\lim_{(\sigma-\rho)\uparrow 0}{\dot  \sigma}\ge-\mu\rho_0$. 
\end{proof}
\begin{lemma}\label{lemma2}
If  the sliding variable  transgresses its lower bound, as given in \eqref{reach}, i.e., $``-\rho(t)"$ then  $(\sigma+\rho)$ will approach $0$ from the right side and the following condition  will hold  \begin{align}
   \lim_{(\sigma+\rho)\downarrow 0}{\dot \sigma}&\le\mu\rho_{0}.\label{taut2}
 \end{align}
\end{lemma}

\begin{proof}
    The proof is similar to that of the Lemma \ref{lemma1}.
\end{proof}
In the next section, we have used the above lemmas to support our proposed control strategy.
\section{Stability Analysis}\label{sec4}
In this section, we have theorized the proposed control strategy  and mathematical proof is presented for the same. 
\begin{theorem}\label{thm1}
Given a system \eqref{sys} holding assumptions \ref{gassump}-\ref{eassum}, if initially sliding variable hold $|\sigma(0)|<\rho_0$ and control law is designed as \eqref{input}, then its trajectory achieve QSM with $\varepsilon$ bandwidth in time $t_r<\frac{1}{\mu}\ln\left(\frac{\rho_0}{\varepsilon-\rho_\infty}\right)$ and system output track the desired trajectory.
\end{theorem}
\begin{proof}
We will prove it by contradiction, and for it, previously established lemmas will be used as a tautology. Let's begin the proof with the statement $``S1"$ below:\par
$S1:$ There exists at least a time instant at which $\sigma$ transgress its bounds, or, $ \exists~ t_j \in  (t_1,\ldots, t_i, \ldots, t_{\bar n} ) \text{~such that~} |\sigma(t_j)|>\rho(t_j),$
 where $t_i<t_{i+1}$, $t_i$ represent $i${th} instant of violation of QSM reaching constraint, $i \in \mathbb{N}$, and $\bar n \in \mathbb{N}$.\\
 Suppose that  \textit{S}$1$ is true. Then we have
\begin{align}\label{t1}
    |\sigma(t)|<\rho(t),~ \forall t\in [0,t_1).
\end{align}
Suppose that at the instant of time $t_1$, the sliding variable is transgressing its bounds. With the following analysis, we will see that it transgresses its bounds.\\
Using \eqref{sys} and \eqref{slide}, we can write 
\begin{align}\label{sigmad}
    \dot \sigma=f(\textbf{x})+g(\textbf{x})u+d-x_{des}^{(n)} +\Sigma_{i=1}^{n-1}c_ie^{(i)}.
\end{align}
For $t<t_1$, using \eqref{t1}, we have $\sigma \in \mathfrak{L}_\infty$, and since $\sigma$ is linear combination of $e^{(i)}, i=\{0,\hdots, n-1\}$, thus $e^{(i)} \in \mathfrak{L}_\infty, i=\{0,\ldots, n-1\}, \forall t\in [0,t_1)$. Further, we have  $e^{(i)}=x_{i+1}-x_{des}^{(i)}, i=\{0,\hdots, n-1\}$ and $x_{des}(t)$ is smooth function, hence $x_{des}(t) \in \mathfrak{L}_\infty, \forall t\in [0,t_1)$, consequently $x_i   \in \mathfrak{L}_\infty, i=\{1, \hdots, n\}, \forall t\in [0,t_1)$. Since $f(\textbf{x})$ is continuous with respect to $\textbf{x}$, i.e., $\textbf{x}=[x_1, \hdots, x_n]$, so $f(\textbf{x}) \in \mathfrak{L}_\infty, \forall t\in [0,t_1)$. Following Assumption \ref{dassump}, we have $d(t) \in \mathfrak{L}_\infty, \forall t\in [0,t_1)$. Further, following Assumption \ref{gassump},  and  \eqref{input}, we  have $\liminf_{(\sigma-\rho)\uparrow0}g(\textbf{x})u=-\infty$ and $\limsup_{(\sigma+\rho)\downarrow0}g(\textbf{x})u=\infty.$ With this analysis we can obtain  $\forall t\in [0,t_1)$ in \eqref{sigmad}
\begin{align}
    \liminf_{(\sigma-\rho)\uparrow0}\dot \sigma&=-\infty,\label{contr1}\\
    \limsup_{(\sigma+\rho)\downarrow0}\dot \sigma&=\infty.\label{contr2}
\end{align}
It can be inferred that \eqref{contr1} and \eqref{contr1} contradict \eqref{taut1} and \eqref{taut2}, respectively. Hence there does not exist any time $t1$, or statement $S1$ is false, or $|\sigma(t)|<\rho(t), \forall t\ge 0.$ Since $\rho(t)=\rho_0\exp{(-\mu t)}+\rho_\infty$, consequently it is straightforward to obtain $|\sigma(t)|<\varepsilon, \forall t\ge\frac{1}{\mu}\ln\left(\frac{\rho_0}{\varepsilon-\rho_\infty}\right).$ This completes the proof for reaching mode. Next, we will analyze QSM.\par
Further, if  $c_i$ in the sliding surface is chosen as $c_i=\binom{n-1}{n-i}a^{n-i}$, then one can find  in chapter 7 of \cite{slotine1991applied} that $|\sigma(t)|<\varepsilon$ implies $|e^{(i)}|<(2a)^i\frac{\varepsilon}{a^{(n-1)}}$, provided $\textbf{e}(0)=\textbf{0}_n$, or if $\textbf{e}(0)\neq\textbf{0}_n$ then $|e^{(i)}|$ will achieve the same bound asymptotically within a shorter time constant $\frac{n-1}{a}$. Thus with proper selection of $\varepsilon$ and $a$, we can achieve asymptotic tracking control, with tracking error in small residual set, i.e., $|e^{(0)}|<\frac{\varepsilon}{a^{(n-1)}}$. This completes the proof.
\end{proof}
\section{Simulation Results}\label{sec5}
To demonstrate the effectiveness of the results, We have considered two examples for the simulation studies, which are as follows.\par
{\textit{Example 1:}} Consider an inverted pendulum system [6] with dynamics given by 
\begin{align*}
    \dot x_1 &= x_2,\\ \dot x_2 &= -\frac{g}{l}\sin(x_1) - \frac{k}{m}x_2 + \sin(x_2) + \frac{1}{ml^2}u + d,\\ y &= x_1,
\end{align*} 
where $x_1$ and $x_2$ represent the angular position and angular velocity, respectively, and $y$ is the output. The system parameters  $m = 0.01$ kg, $l = 9.8$ m, $k = 0.01$, and $g$ represent the end mass, length of the rod, friction coefficient, and acceleration due to gravity, respectively, with $u$ as the control input. Additionally, a disturbance $d = 0.5\sin(t)$ is applied to the system. The desired output is chosen as $y_d = \sin(t)$. To assess the effectiveness of the proposed algorithm, we conducted simulations with various initial conditions of the states. It is chosen as  $\textbf{x}(0)=[0.9~ 0.9]$, $\textbf{x}(0)=[0.7~ 0.7]$,  $\textbf{x}(0)=[0.3~ 0.3]$, and $\textbf{x}(0)=[0.1~ 0.1]$.\par
\begin{figure}
    \centering
\includegraphics[width=9cm,height=3cm]{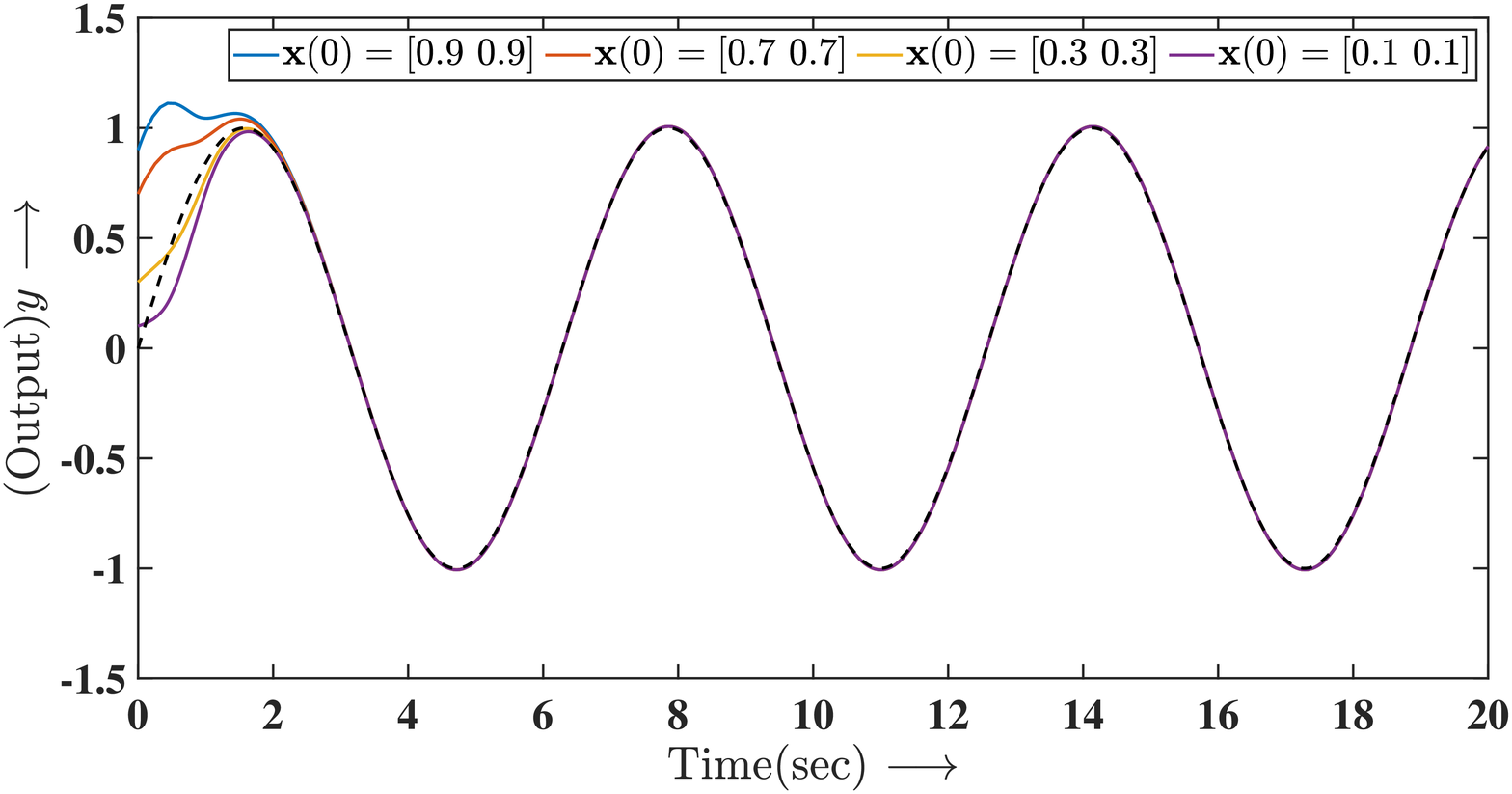}
    \caption{Tracking desired output (Example 1); Solid and dashed lines represent the actual output and desired trajectory, respectively. }
    \label{figop}
\end{figure}
\begin{figure}
    \centering
\includegraphics[width=9cm,height=3cm]{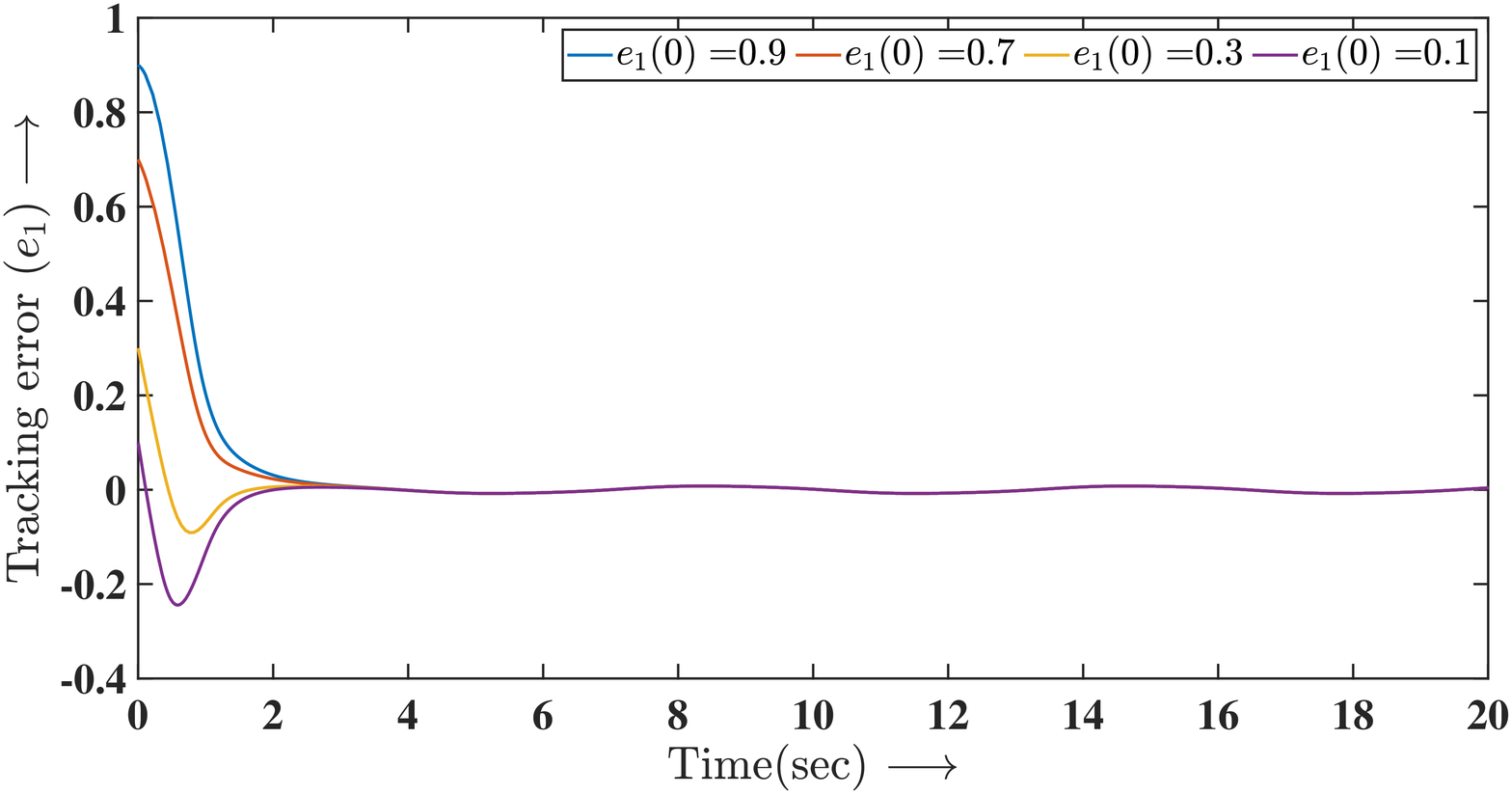}
    \caption{Tracking error performance (Example 1).}
    \label{figerr}
\end{figure}

\begin{figure}
    \centering
\includegraphics[width=9cm,height=3cm]{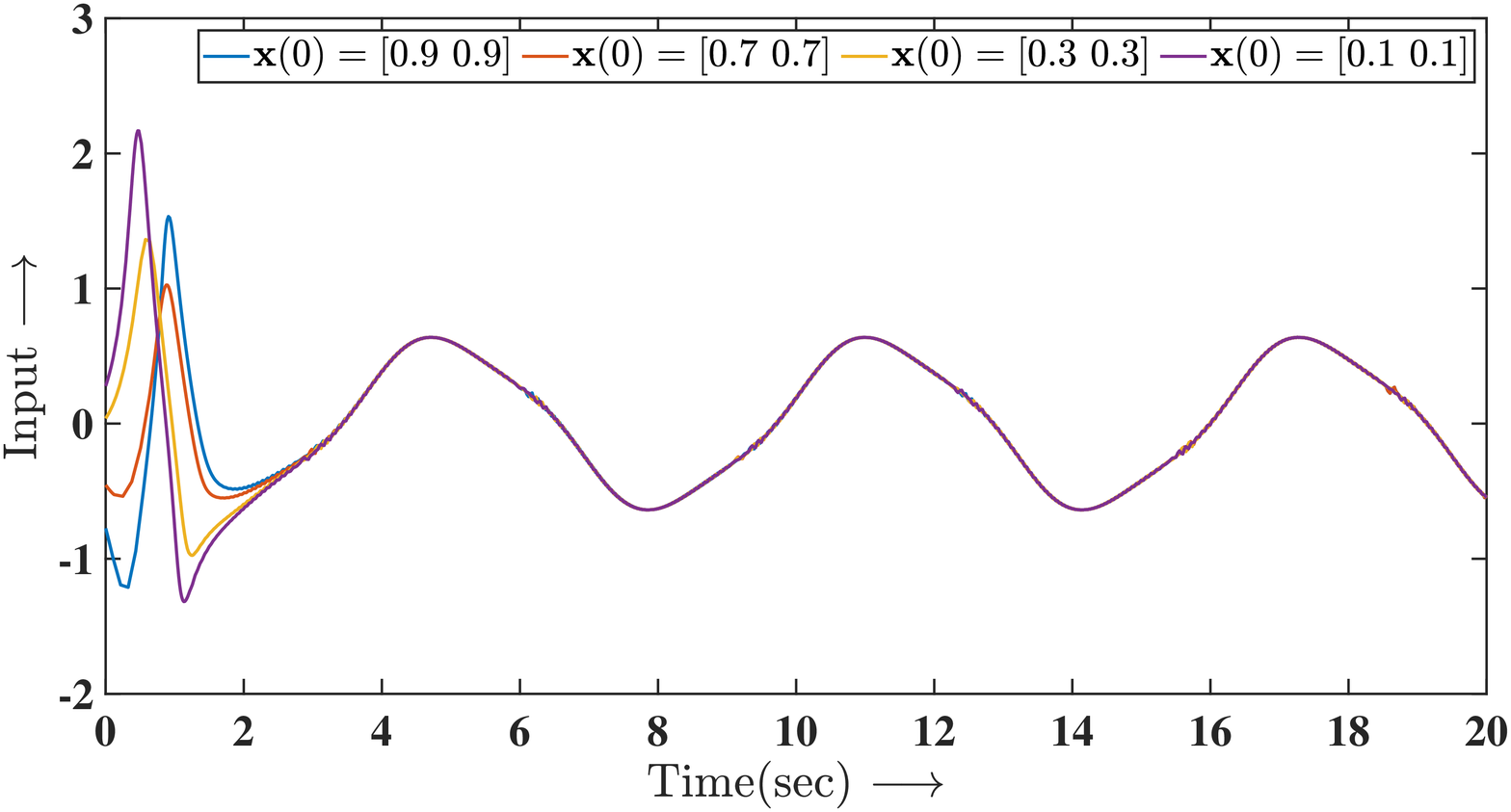}
    \caption{Input performance (Example 1).}
    \label{figip}
\end{figure}

\begin{figure}
    \centering
\includegraphics[width=9cm,height=3cm]{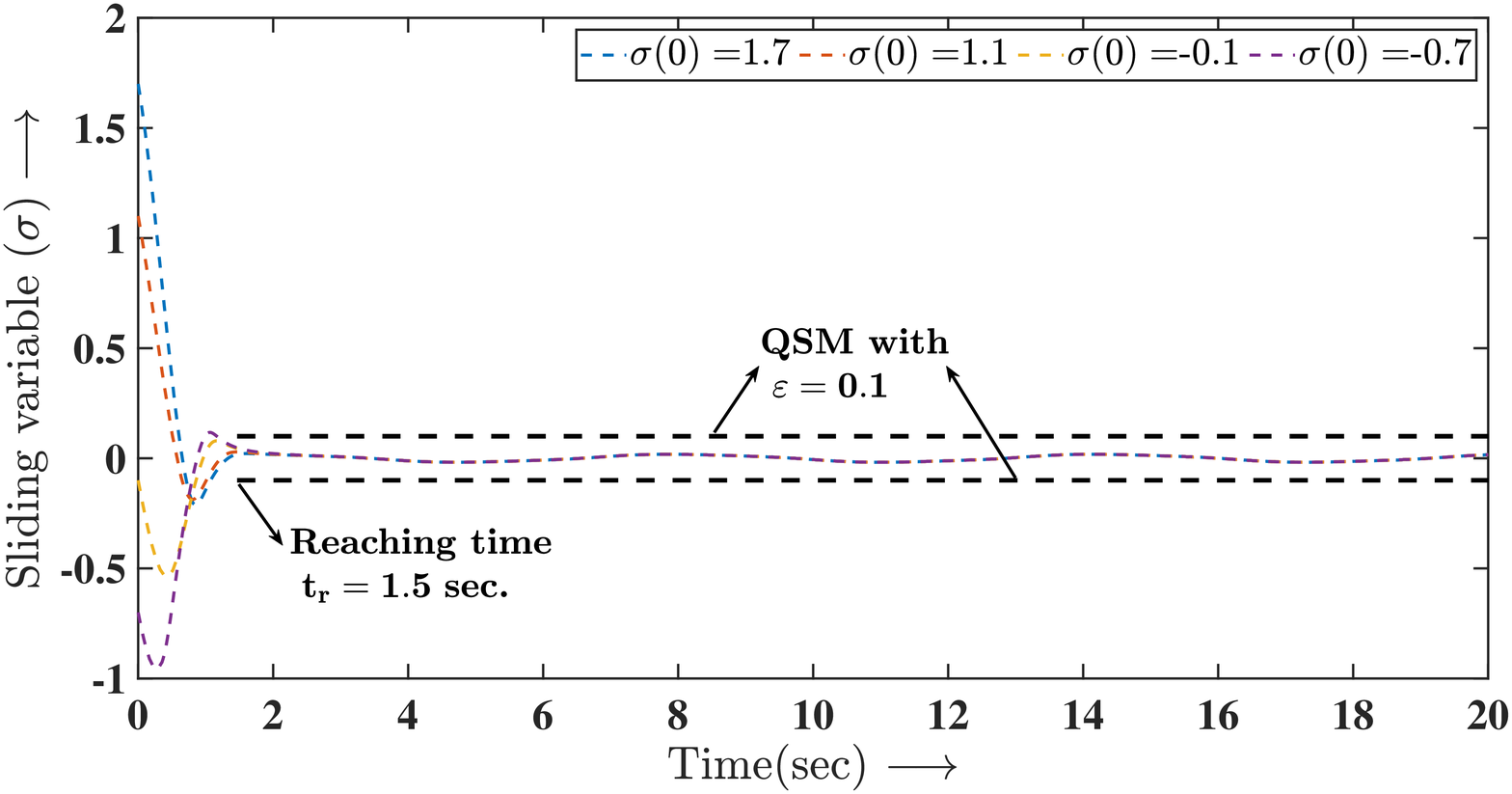}
    \caption{QSM reaching law performance (Example 1). }
    \label{figsv}
\end{figure}
QSMC is designed as follows. Firstly, the sliding surface is chosen as given in \eqref{slide}, with $c_1=2$, resulting in $\sigma=2e^{(0)}+e^{(1)}$, where $e^{(0)}= y - y_d$ and $e^{(1)}=x_2-\dot y_d$. Secondly, a reaching law for QSMC is designed with a small bandwidth $\varepsilon$, set to $\varepsilon=0.1$. The reaching law is given by \eqref{reach}, i.e., $|\sigma(t)|<\rho(t)$, where $\rho(t)=4\exp(-3t)+0.05$. It can be verified that $\rho(t)$ satisfies conditions $C1$ and $C2$ for all sets of chosen initial conditions. The upper bound for the reaching time can be computed as $t_r<\frac{1}{3}\ln{\frac{4}{0.1-0.05}}$, which yields $t_r<1.5$ seconds. Finally, in accordance with the chosen reaching law, to achieve QSMC, i.e., $|\sigma(t)|<\varepsilon$ for all $t\ge1.5$, the input is applied as given in \eqref{input}. 
The simulation results are presented in figures \ref{figop}-\ref{figsv}. In Fig. \ref{figop}, it can be observed that the output accurately tracks the desired output for all sets of initial conditions, demonstrating satisfactory tracking error performance as depicted in Fig. \ref{figerr}. Additionally, Fig. \ref{figip} illustrates that the input remains bounded. Moreover, Fig. \ref{figsv} demonstrates that the sliding variable reaches the QSM within the calculated time of $t_r<1.5$ seconds for all sets of initial conditions. These simulation results are consistent with the Theorem \ref{thm1}. \par
\textit{Example 2:} Consider a fourth-order numerical system described as 
\begin{equation}\label{ex2}
\begin{split}
   \dot x_i&=x_{i+1}, ~i=1,2,3,\\
    \dot x_4&=x_3\sin{(2x_2)}+x_1\cos{x_4}+(2-\sin{x_4})u+d, \\
    y&=x_1,
\end{split}   
\end{equation}
\begin{figure}
    \centering
\includegraphics[width=9cm,height=3cm]{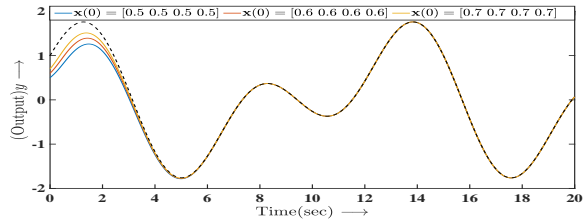}
    \caption{Tracking desired output (Example 2); Solid and dashed lines represent the actual output and desired trajectory, respectively.}
    \label{figop2}
\end{figure}
\begin{figure}
    \centering
\includegraphics[width=9cm,height=3cm]{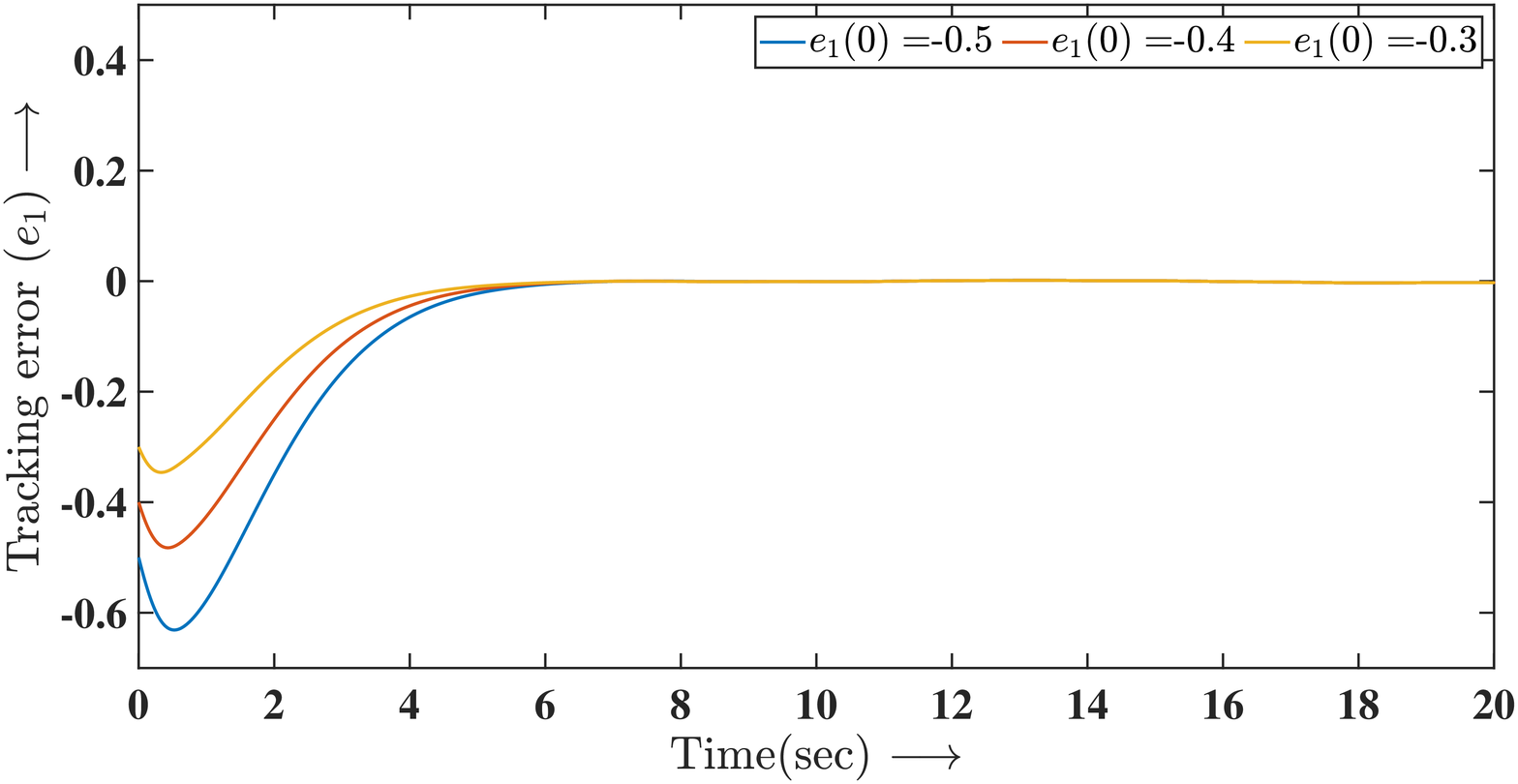}
    \caption{Tracking error performance (Example 2).}
    \label{figerr2}
\end{figure}
\begin{figure}
    \centering
\includegraphics[width=9cm,height=3cm]{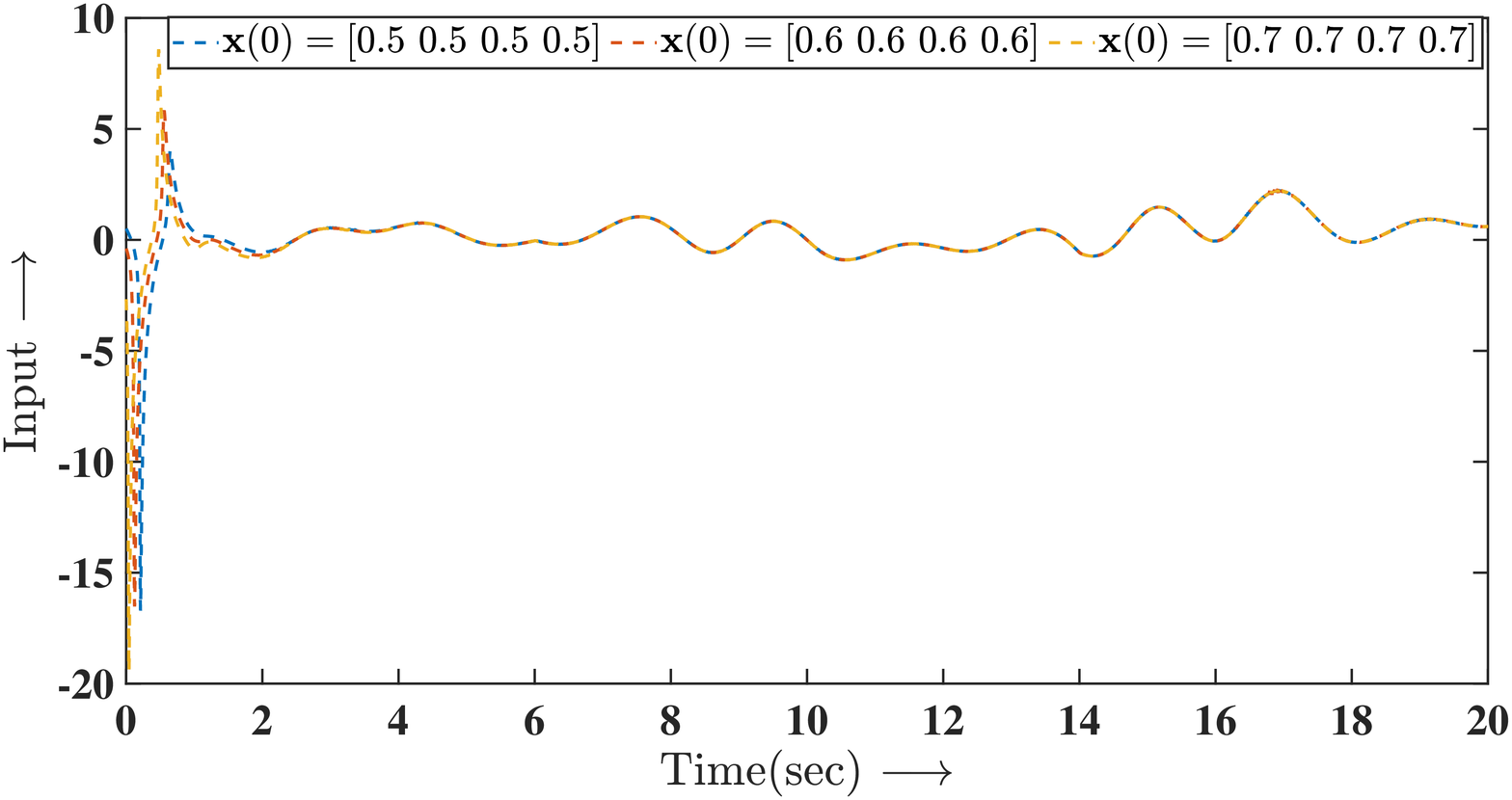}
    \caption{Input performance (Example 2).}
    \label{figip2}
\end{figure}
\begin{figure}[h]
    \centering
\includegraphics[width=9cm,height=3cm]{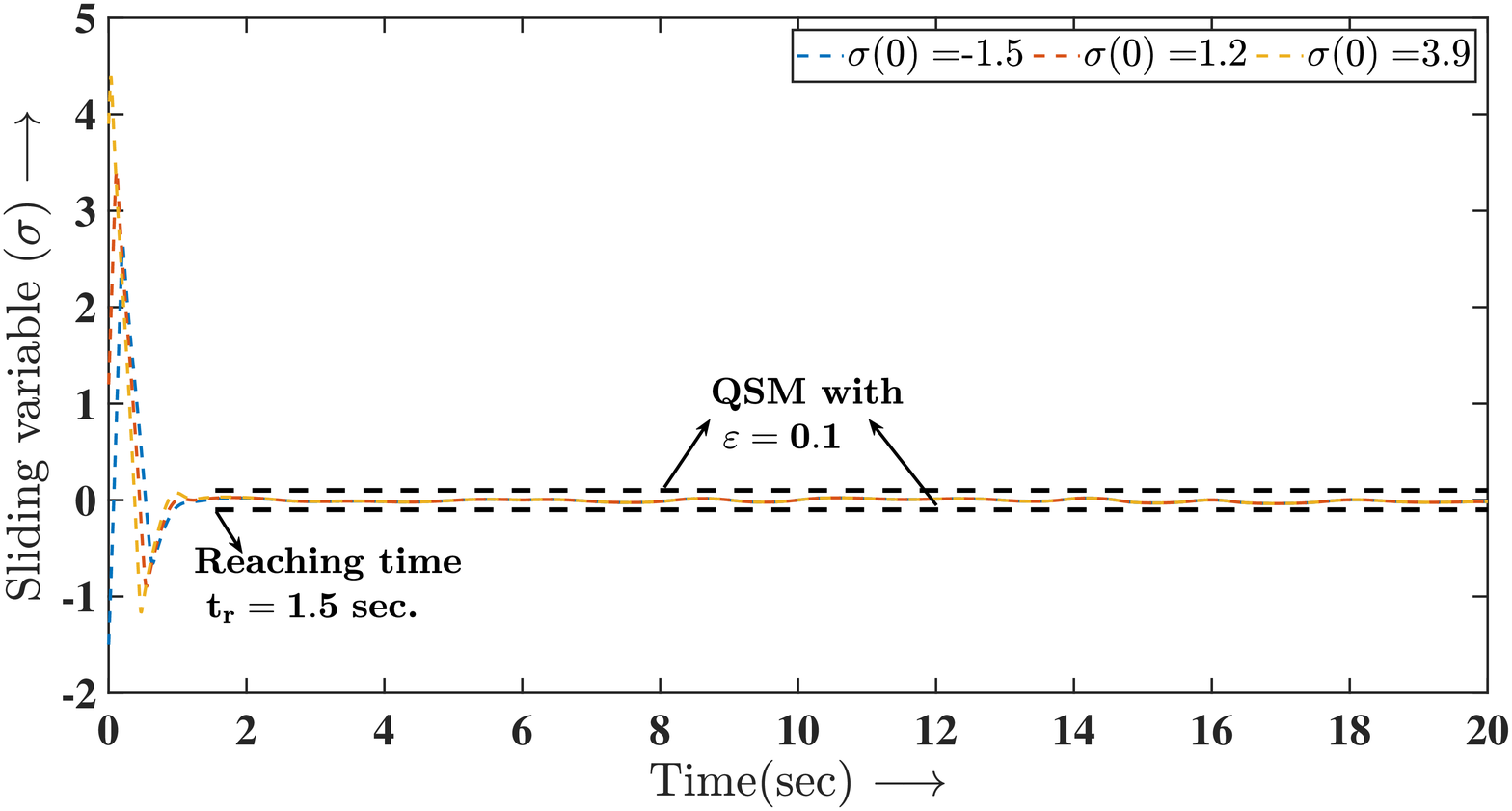}
    \caption{QSM reaching law performance (Example 2). }
    \label{figsv2}
\end{figure}
where $x_i, ~i=1,\ldots,4$ are states, $u$ and $y$ are  input and output of the system \eqref{ex2}, respectively. To verify the robustness of the proposed algorithm, the disturbance $d(t)$ is considered as $$d(t)=\begin{cases}
    0.5\sin{\frac{\pi}{2}t}, ~0\le t\le 6,\\
    \sin{\pi t}, ~6\le t\le 9,\\
    \cos{\pi t}-1,  ~t\ge 9.
\end{cases}$$ The desired output is chosen as $y_d=\sin{t}+\cos{0.5t}.$  For the simulation, three different initial conditions are chosen, $\textbf{x}(0)=[0.5~0.5~0.5~0.5], \textbf{x}(0)=[0.6~0.6~0.6~0.6]$ and $ \textbf{x}(0)=[0.7~0.7~0.7~0.7].$ The QSMC is designed using \eqref{input} to reach QSM of bandwidth, $\varepsilon=0.1$. The steps for the design are similar to the previous example, and it involves the following parameters: $c_1=6, c_2=12, c_3=8$ and $\rho(t)=5\exp(-3t)+0.05$ for the design of sliding surface and reaching law, respectively. The upper bound of reaching time for QSM of $\varepsilon=0.1$ bandwidth can be computed as $t_r<\frac{1}{3}\ln{\frac{5}{0.1-0.05}}$, which yields $t_r<1.53$. The simulation results are shown in Figs. \ref{figop2}-\ref{figsv2}. Fig. \ref{figsv2} shows that the sliding variable reaches the QSM in the estimated time for all sets of chosen state initial conditions. Fig. \ref{figop2}  shows the tracking of the desired output. Further, Fig. \ref{figerr2} and  Fig. \ref{figerr2} demonstrate satisfactory tracking performance and boundedness of input, respectively. Similar to the previous example, the simulation results are consistent with Theorem \ref{thm1}.

\section{Conclusion}\label{sec6}
The paper presents a control approach for achieving QSM and tracking control for continuous-time systems with uncertainties and disturbances. The proposed QSMC is chattering-free, does not need any knowledge of the system, and guarantees that the system trajectory reaches the QSM in a fixed time. The stability of the system is analyzed using mathematical proofs and simulation results. Further research can be conducted to apply the proposed approach to various classes of nonlinear systems and validate its performance in practical applications.
\bibliographystyle{ieeetr}
\bibliography{ref1.bib}

\begin{thebibliography}{10}

\bibitem{smcbook}
C.~Edwards and S.~Spurgeon, {\em Sliding mode control: theory and
  applications}.
\newblock Crc Press, 1998.

\bibitem{FENG20141310}
Y.~Feng, F.~Han, and X.~Yu, ``Chattering free full-order sliding-mode
  control,'' {\em Automatica}, vol.~50, no.~4, pp.~1310--1314, 2014.

\bibitem{DU201687}
H.~Du, X.~Yu, M.~Z. Chen, and S.~Li, ``Chattering-free discrete-time sliding
  mode control,'' {\em Automatica}, vol.~68, pp.~87--91, 2016.

\bibitem{8064665}
Y.~Pan, C.~Yang, L.~Pan, and H.~Yu, ``Integral sliding mode control:
  Performance, modification, and improvement,'' {\em IEEE Transactions on
  Industrial Informatics}, vol.~14, no.~7, pp.~3087--3096, 2018.

\bibitem{880629}
G.~Bartolini, A.~Ferrara, E.~Usai, and V.~Utkin, ``On multi-input
  chattering-free second-order sliding mode control,'' {\em IEEE Transactions
  on Automatic Control}, vol.~45, no.~9, pp.~1711--1717, 2000.

\bibitem{7913590}
J.~Lee, P.~H. Chang, and M.~Jin, ``Adaptive integral sliding mode control with
  time-delay estimation for robot manipulators,'' {\em IEEE Transactions on
  Industrial Electronics}, vol.~64, no.~8, pp.~6796--6804, 2017.

\bibitem{1263584}
C.-H. Tsai, H.-Y. Chung, and F.-M. Yu, ``Neuro-sliding mode control with its
  applications to seesaw systems,'' {\em IEEE Transactions on Neural Networks},
  vol.~15, no.~1, pp.~124--134, 2004.

\bibitem{8744495}
Y.~Chu, J.~Fei, and S.~Hou, ``Adaptive global sliding-mode control for dynamic
  systems using double hidden layer recurrent neural network structure,'' {\em
  IEEE Transactions on Neural Networks and Learning Systems}, vol.~31, no.~4,
  pp.~1297--1309, 2020.

\bibitem{9942821}
C.~Jia, D.~Kong, and L.~Du, ``Recursive terminal sliding-mode control method
  for nonlinear system based on double hidden layer fuzzy emotional recurrent
  neural network,'' {\em IEEE Access}, vol.~10, pp.~118012--118023, 2022.

\bibitem{8818641}
K.~Shao, J.~Zheng, K.~Huang, H.~Wang, Z.~Man, and M.~Fu, ``Finite-time control
  of a linear motor positioner using adaptive recursive terminal sliding
  mode,'' {\em IEEE Transactions on Industrial Electronics}, vol.~67, no.~8,
  pp.~6659--6668, 2020.

\bibitem{10004950}
P.~K. Mishra and P.~Jagtap, ``Approximation-free prescribed performance control
  with prescribed input constraints,'' {\em IEEE Control Systems Letters},
  vol.~7, pp.~1261--1266, 2023.

\bibitem{10093463}
P.~K. Mishra and P.~Jagtap, ``On controller design for unknown nonlinear
  systems with prescribed performance and input constraints,'' in {\em 2022
  Eighth Indian Control Conference (ICC)}, pp.~212--217, 2022.

\bibitem{370376}
W.~Gao, Y.~Wang, and A.~Homaifa, ``Discrete-time variable structure control
  systems,'' {\em IEEE Transactions on Industrial Electronics}, vol.~42, no.~2,
  pp.~117--122, 1995.

\bibitem{704892}
A.~Bartoszewicz, ``Discrete-time quasi-sliding-mode control strategies,'' {\em
  IEEE Transactions on Industrial Electronics}, vol.~45, no.~4, pp.~633--637,
  1998.

\bibitem{4132912}
S.~Janardhanan and B.~Bandyopadhyay, ``Multirate output feedback based robust
  quasi-sliding mode control of discrete-time systems,'' {\em IEEE Transactions
  on Automatic Control}, vol.~52, no.~3, pp.~499--503, 2007.

\bibitem{4105940}
X.-D. Li, T.~W.~S. Chow, J.~K.~L. Ho, and H.-Z. Tan, ``Repetitive learning
  control of nonlinear continuous-time systems using quasi-sliding mode,'' {\em
  IEEE Transactions on Control Systems Technology}, vol.~15, no.~2,
  pp.~369--374, 2007.

\bibitem{LI2009103}
X.-D. Li, T.~W. Chow, and J.~K. Ho, ``Quasi-sliding mode based repetitive
  control for nonlinear continuous-time systems with rejection of periodic
  disturbances,'' {\em Automatica}, vol.~45, no.~1, pp.~103--108, 2009.

\bibitem{1102073}
S.~Gutman, ``Uncertain dynamical systems--a lyapunov min-max approach,'' {\em
  IEEE Transactions on Automatic Control}, vol.~24, no.~3, pp.~437--443, 1979.

\bibitem{utkin2017sliding}
V.~Utkin, J.~Guldner, and J.~Shi, {\em Sliding mode control in
  electro-mechanical systems}.
\newblock CRC press, 2017.

\bibitem{1532412}
A.~Levant, ``Quasi-continuous high-order sliding-mode controllers,'' {\em IEEE
  Transactions on Automatic Control}, vol.~50, no.~11, pp.~1812--1816, 2005.

\bibitem{8123916}
A.~Levant and B.~Shustin, ``Quasi-continuous mimo sliding-mode control,'' {\em
  IEEE Transactions on Automatic Control}, vol.~63, no.~9, pp.~3068--3074,
  2018.

\bibitem{8734866}
L.~Zhou, J.~She, X.-M. Zhang, Z.~Cao, and Z.~Zhang, ``Performance enhancement
  of rcs and application to tracking control of chuck-workpiece systems,'' {\em
  IEEE Transactions on Industrial Electronics}, vol.~67, no.~5, pp.~4056--4065,
  2020.

\bibitem{CHOI1994899}
S.-B. Choi, D.-W. Park, and S.~Jayasuriya, ``A time-varying sliding surface for
  fast and robust tracking control of second-order uncertain systems,'' {\em
  Automatica}, vol.~30, no.~5, pp.~899--904, 1994.

\bibitem{bartoszewicz1996time}
A.~Bartoszewicz, ``Time-varying sliding modes for second-order systems,'' {\em
  IEE Proceedings-Control Theory and Applications}, vol.~143, no.~5,
  pp.~455--462, 1996.

\bibitem{slotine1991applied}
J.-J.~E. Slotine, W.~Li, {\em et~al.}, {\em Applied nonlinear control},
  vol.~199.
\newblock Prentice hall Englewood Cliffs, NJ, 1991.

\end{thebibliography}
\end{document}